\documentclass[prx,amsmath,amssymb,floatfix, twocolumn, superscriptaddress, notitlepage]{revtex4-2} %onecolumn,
\usepackage{wrapfig}
\usepackage{graphicx}
\usepackage{float}
\usepackage{amssymb,amsthm,amsfonts,amstext,amsmath}
\usepackage{url}
\usepackage[colorlinks=true,citecolor=cyan,urlcolor=magenta]{hyperref}

\usepackage[dvipsnames]{xcolor}
\usepackage{tcolorbox}
\usepackage{enumitem}
\usepackage{tikz}

\usepackage{cleveref}
\crefname{equation}{Eq.}{Eqs.}
\crefname{section}{Sec.}{Sections}
\crefname{figure}{Fig.}{Figs.}
\def\be{\begin{equation}}
\def\ee{\end{equation}}
\def\bea{\begin{eqnarray}}
\def\eea{\end{eqnarray}}
\def\bma{\begin{mathletters}}
\def\ema{\end{mathletters}}

\def\q0{\underline{0}}

\def\S{{\cal S}}

\def\id{{\mathbb I}}

\def\R{\mathbb{R}}

\def\tr{\mbox{tr}}

\def\one{\leavevmode\hbox{\small1\normalsize\kern-.33em1}}

\def\bra#1{\langle#1|} \def\ket#1{|#1\rangle}

\def\proj#1{\ket{#1}\!\bra{#1}}

\newtheorem{theo}{Theorem}

\newtheorem{defin}[theo]{Definition}

\newtheorem{prop}[theo]{Proposition}

\newtheorem{problem}{Problem}
\def\id{{\mathbb I}}

\begin{document}
\title{Memory attacks in network nonlocality and self-testing}

\author{Mirjam Weilenmann}
\affiliation{Inria, Télécom Paris - LTCI, Institut Polytechnique de Paris, 91120 Palaiseau, France}
\email{mirjam.weilenmann@inria.fr}
\author{Costantino Budroni}
\affiliation{Department of Physics ``E. Fermi'' University of Pisa, Largo B. Pontecorvo 3, 56127 Pisa, Italy}
\email{costantino.budroni@unipi.it}
\author{Miguel Navascués}
\affiliation{Institute for Quantum Optics and Quantum Information--IQOQI Vienna, Austrian Academy of Sciences, Boltzmanngasse 3, 1090 Vienna, Austria}
\email{miguel.navascues@oeaw.ac.at}

\begin{abstract}
We study what can or cannot be certified in communication scenarios where the assumption of independence and identical distribution (iid) between experimental rounds fails. In this respect, we prove that membership tests for non-convex sets of correlations cannot be formulated in the non-iid regime. Similarly, it is impossible to self-test non-extreme quantum operations, such as mixed states, or noisy quantum measurements, unless one allows more than a single use thereof within the same experimental round. One consequence of our results is that non-classicality in causal networks without inputs cannot be experimentally demonstrated. By analyzing optimal non-iid strategies in the triangle scenario, we raise the need to take into account the prior communication required to set up a causal network.

\end{abstract}

\maketitle

\section{Introduction}

Hypothesis testing plays an important role in experiments for certifying properties of quantum systems.
Statistical tools based on martingale theory or prediction-based-ratio analysis  have been successfully applied~\cite{Gill1, Gill2, ZhangPRA2011, ZhangPRA2013}, especially in the context of Bell inequality violations \cite{Giustina_2015, ShalmPRL2015, storz}, where it is common to work in an adversarial setting. These techniques allow us to formulate binary hypothesis tests that are secure against adversaries that introduce variations between the different rounds of a protocol. As a consequence, these tests are valid beyond the {\it independent and identically distributed} (iid)~regime and take into account the possibility of {\it memory attacks}~\cite{memory_attack}.

In this work, we study the consequences of dropping the iid assumption in two standard certification tasks: (a) proving that the experimental devices can generate measurement statistics outside of a given set of correlations or (b) close to a given distribution. The first task is the basis of non-classicality detection, understood as the certification of quantum (or at least non-classical) effects within the experimental setup. The second task is instrumental for self-testing \cite{tsirel, Mayers_Yao}: certifying that the generated statistics are close enough to a given distribution, it is sometimes possible to make claims about the inner workings of the devices involved in the experiment.

Regarding the first task, we conclude that membership tests for non-convex sets of correlations cannot be formulated in the non-iid\ regime. This implies that, for general causal networks, there are correlations for which one cannot certify incompatibility with the network. In particular, no causal network without inputs, such as the famous triangle scenario~\cite{Steudel, Fritz}, allows certifying non-classicality of the observed correlations. Another consequence is that, in general, one cannot certify that a multipartite source is distributing separable states, due to the non-convexity of the set of (highly) entangled states.

Regarding the second task, we prove that, in the non-iid regime, robust self-testing is only possible if the considered distribution is an extreme point of the set of physical correlations. Thus it is impossible to self-test mixed quantum states and non-extreme quantum measurements or instruments, independently of the communication scenario, unless one allows multiple identical preparations or uses thereof in a single experimental round. This means that the constructions presented in \cite{Nikolai, sarkar2023universal} to self-test arbitrary quantum states and measurements cannot be extended to non-iid scenarios.

We conclude our paper by analysing non-iid strategies that separate players in a network might follow to win a membership test devised for the iid case. We observe that many such strategies, while formally compatible with the network configuration, would require violating the network's causal constraints to be set up. Hence we raise the need to discuss how prior communication might compromise the conclusions of a causal experiment.

\section{The input/output scenario}
\label{sec:input_output}
We consider a scenario where the experimenters have access to a device with an input and an output port. In every experimental round, the experimenters feed a symbol $x\in\{1,...,X\}$ into the input port and read a symbol $a\in\{1,...,A\}$ on the output. Inputs and outputs might be composite, e.g.\ in a bipartite Bell experiment, the input is of the form $(x,x')$ and the output,  $(a,a')$. In some cases, the input may also be trivial, i.e., take only one possible value. This input/output framework thus encompasses Bell tests, prepare-and-measure scenarios, and causal network experiments.  

In general, the probability of obtaining output $a$ given input $x$ in round $k$ depends on both $k$ and the past history of inputs and outputs. Invoking Bayes' rule, the most general statistical behavior of a device in an input/output experiment is given by the marginals 
\begin{align}
P_1(a_1|x_1), P_2(a_2|x_2, s_1),...,P_n(a_n|x_n, s_{n-1}),
\end{align}
where $s_k=(a_1,x_1,...,a_{k},x_{k})$ represents the list of inputs and outputs occurring prior to or in round $k$. We will also denote by $P^{(k)}$ the prior distribution  of the inputs and outputs in the first $k$ rounds. A special type of behavior arises when, in every round $k$, the probability of obtaining output $a_k$ is solely dependent on $k$ and the input $x_k$, i.e., when $P_k(a_k|x_k, s_{k-1})=P_k(a_k|x_k)$, for $k=1,...,n$. We will then denote $P^{(n)}$ by $\bigotimes_{k=1}^nP_k$. If $P_k=P,\forall k$, we say that the device's behavior is iid. 
In that case, we write $P^{(n)}=P^{\otimes n}$.

If we posit some physical model $S$ generating $P_k(a_k|x_k, s_{k-1})$, we might find that the corresponding set of correlations $\S=\{P(a|x):P(a|x) \mbox{ compatible with } S\}$ does not exhaust the set of all possible probability distributions. In particular, we might find a constraint of the form 
\begin{equation}\label{eq:prob_ineq}
F(P)\leq \alpha,\forall P\in \S,
\end{equation}
where $F$ is a function of the distribution $P(a|x)$ and $\alpha\in\R$, which is violated for some $P\notin \S$. Examples of such constraints are Bell inequalities (when $\S$ is the set of local realistic correlations \cite{Brunner2014, scarani2019bell}) or prepare-and-measure dimension witnesses (when the experimental outcomes are generated by conducting measurements labeled by $y$ on a $D$-dimensional quantum system in state $x$ \cite{Gallego_2010}).

Now, suppose that we wish to refute the (null) hypothesis that model $S$ describes the device's behavior. Within the input/output framework, this is equivalent to proving that the device can generate a distribution $P\not\in \S$. The most general way to proceed is by defining a binary hypothesis test.

\begin{defin}{\textbf{Binary Hypothesis Test.}}\label{def:hyp_test}
An $n$-round binary hypothesis test is a prescription  to decide what symbol $x_k$ to input into the device in each experimental round $k$, as well as how to decide on the final outcome of the test: $0$ (signifying that the test aiming to reject the null hypothesis failed) or $1$ (meaning that the null hypothesis is rejected).
\end{defin}

According to this definition, the most general statistical test $T$ for an input/output device is defined in terms of the marginal distributions $Q_1(x_1), Q_2(x_2|s_1), ..., Q_n(x_n|s_{n-1})$, which generate the next input based on the past history of inputs and outputs, and the distribution $Q_{f}(o|s_{n})$, with $o\in\{0,1\}$, used to generate the final, binary outcome of the test. Given the $n$-round behavior $P^{(n)}$ of the experimental device, the probability to declare the null hypothesis falsified is thus
\begin{align}\label{eq:PT}
P_{T}(o=1|P^{(n)}):=&\sum_{s_n}P^{(n)}(s_n) Q_1(x_1)Q_2(x_2|s_1)...\nonumber\\
&Q_n(x_n|s_{n-1})Q_{f}(1|s_n).
\end{align}

To see that Eq.~\eqref{eq:PT}  captures the most general hypothesis test, note that the above does not only apply to hypothesis tests where the total number of rounds $n$ is fixed \emph{a priori}, but also to hypothesis tests where the total number of rounds dynamically depends on the sequence of outputs of the experimental device. Such variable-$n$ protocols are indeed included in the formulation above, as long as they have a maximum number $N$ of rounds (this is a realistic assumption, since there is no such thing as an infinitely long experiment). In that case, the protocol would be defined through marginals $Q_1,...,Q_N$ and the condition of variable $n$ can be modeled by making the outcome distribution $Q_{f}(o|s_N)$ to be of the form $Q_{f}(o|s_N)=Q_{f}(o|s_n)$, for some lists of inputs/outputs $s_n$. Namely, the device would still be probed after the $n^{th}$ round, but the subsequent results would not be taken into account. Alternatively, one can define a new input $x=\emptyset$, signifying that the device is not to be measured in that round, and set the marginals $Q_1,...,Q_{n}$ defining the test so that, if the data $s_k$ collected so far contains any input labeled $\emptyset$, then the output of $Q_k$ must be $\emptyset$ with certainty.

Suppose, then, that we wish to prove that our device can produce distributions $P(a|x)$ such that $F(P)> \alpha$, with $F$ being a continuous function. Equivalently, we wish to refute the hypothesis that, in every round, the device generates a distribution $P_k$ satisfying $F(P_k)\leq \alpha$.

In the iid\ scenario, all we need to do is input different values of $x_1,...,x_n$ and note the outcomes $a_1,...,a_n$. Then, we can approximate the probabilities $P(a|x)$ by the observed frequencies, i.e.,
\begin{equation}\label{eq:prob_est_iid}
\tilde{P}(a|x):=\frac{|\{a_i=a, x_i=x\}|}{|\{x_i=x\}|},
\end{equation}
where $|\{a_i=a, x_i=x\}|$ is the number of times $x_i=x$ and $a_i=a$ is observed.
Finally, we evaluate $F(\tilde{P})$, corresponding to the average value of $F$ in our experiment, and its standard deviation $\sigma$. We then declare the null hypothesis rejected if the result is larger than $\alpha+ K \sigma$ for some predetermined $K>0$. In other words, to declare the null hypothesis rejected we require that the average is $K$ standard deviations away from the bound.

The instructions above are in effect an $n$-round binary hypothesis test, with the property that, for large $n$, the outcome will be $0$ with high probability if the device cannot generate distributions beyond $\S$. On the contrary, if the device does generate distributions outside of $\S$, then the test will output $1$ with high probability. 

In the iid\ scenario, it is similarly easy to carry out a \emph{self-testing} experiment. Let $\|\bullet\|$ be an arbitrary norm and suppose that $P\in \S$ is such that, for all $Q\in \S$, the relation $\|P-Q\|<\varepsilon$ implies that the physical realization of $Q$ within theory $S$ requires the device to prepare states and/or conduct operations close to some reference ones. Then we say that $P$ \emph{robustly self-tests} those states and/or operations.

What happens, though, if we drop the iid\ assumption? As shown in~\cite{ZhangPRA2013}, if $F$ is linear on the probabilities $P$, then one can design a family of hypothesis tests $(T^n)_n$ with the property that
\begin{align}
&\lim_{n\to\infty}P_{T^n}(1|P^{\otimes n})=1,\forall P \text{ s.t. } F(P(a|x))>\alpha \nonumber\\
&P_{T^n}(1|P^{(n)})\leq \varepsilon_n,\forall P^{(n)} \text{ s.t. } F(P_k(a|x,s_k))\leq\alpha,\forall k,s_k,
\label{holy_grail}
\end{align}
for some $\{\epsilon_n\}_n\subset\R^+$ with $\lim_{n\to\infty}\epsilon_n=0$ (in the statistics literature, $\epsilon_n$ is known as the \emph{$p$ value} of test $T^n$).

Building on this observation, for any convex set $\S$ of correlations the authors of~\cite{ZhangPRA2013} provide a hypothesis test to falsify the null hypothesis that $P_k(a_k|x_k,s_k)\in \S \ \forall k$.

There exist, however, interesting sets of correlations that are not convex. These sets appear, for instance, in the study of causal networks~\cite{Fritz}. To partially characterize these sets, non-linear functions such as polynomials or linear combinations of entropies have been proposed~\cite{Chaves_poly, Inflation, entropy-review, Chaves}. As we saw, verifying that a system violates a non-linear constraint in the iid\ case is a trivial endeavor. But, how about the non-iid\ case? Can one extend the results of~\cite{ZhangPRA2013} to non-linear functions? 

One can ask the same questions regarding self-testing. In \cite{Nikolai}, Miklin and Oszmaniec propose a scheme, relying on the iid assumption, to self-test arbitrary quantum states and Positive Operator Valued Measures (POVMs) in prepare-and-measure scenarios. In particular, they can self-test states and measurements giving rise to single-shot input/output distributions $P$ that are not extreme in the set of $2$-dimensional quantum realizations. Can this scheme be extended to the non-iid regime?

\section{Testing non-convex sets}
Let us consider a set of correlations $\S$. If the null hypothesis holds, namely, if $P_k(a|x,s_k)\in \S,\forall k, s_k$, then we wish our hypothesis test to output $1$ with very low probability. At the same time, we wish the test to be meaningful, i.e., for some $P\not\in \S$ we wish that the probability of detection of $P^{\otimes n}$ be close to $1$. This defines the following problem.

\begin{problem}[Membership Problem] \label{prob:stat_test_design}
    Let $\mathcal{S}$ be a set of correlations. For $P \notin \mathcal{S}$, we aim to develop a family of hypothesis tests $(T^n)_n$, s.t. 
    \begin{equation}
     \lim_{n \rightarrow \infty} P_{T^n}(1| P^{\otimes n}) = 1.   
     \label{non_local}
    \end{equation}
    In addition, for all $P^{(n)}$ with $P_k(a_k|x_k,s_k)\in \S,\forall k,s_k$,
    \begin{align}
    & P_{T^n}(1| P^{(n)}) \leq \varepsilon_n,\nonumber\\
    &\lim_{n\to\infty}\varepsilon_n=0.
    \label{eq:local_bound}
    \end{align}
\end{problem}

Now, if $P\not\in \mbox{conv}(\S)$, where $\mbox{conv}(\cdot)$ denotes the convex hull, then we can use the results of~\cite{ZhangPRA2013} to devise a test to disprove the (stronger) null hypothesis  
$P_k\in \mbox{conv}(\S),\forall k, s_k$.
The following proposition shows that it is \emph{not} possible to devise a hypothesis test that rejects the null hypothesis and yet detects $P\not\in\S$, $P\in\mbox{conv}(\S)$.

\begin{prop} \label{theo:main}
    There does not exist a solution to Problem~\ref{prob:stat_test_design} for $P \in \operatorname{conv}(\mathcal{S}) \setminus \mathcal{S}$. 
\end{prop}

\begin{proof}
    By contradiction. Let $(T^n)_n$ be a family of hypothesis tests satisfying Eqs.~\eqref{non_local} and (\ref{eq:local_bound}). Since $P \in \operatorname{conv}(\mathcal{S})$, we consider any decomposition of $P= \sum_i \lambda_i P_i $ as a convex combination of distributions $\{P_i\}_i\subset\S$. Then we have that
    \begin{align}
    P_{T^n}(1 | P^{\otimes n})=& \sum_{i_1, \ldots i_n} \lambda_{i_1} \cdots \lambda_{i_n} P(1| P_{i_1} \otimes \cdots \otimes P_{i_n}) \nonumber\\
    \leq & \ \varepsilon_n \sum_{i_1, \ldots i_n} \lambda_{i_1} \cdots \lambda_{i_n} = \varepsilon_n ,
    \end{align}
where we used first the linearity of $P_{T^n}$ with respect to $P^{(n)}$, see Eq.~\eqref{eq:PT}, and then Eq.~\eqref{eq:local_bound}. It follows that, $\lim_{n\to\infty}P_{T^n}(1|P^{\otimes n})=0$, contradicting Eq.~\eqref{non_local}.
\end{proof}

Proposition \ref{theo:main} implies that the membership tests for non-convex sets currently used in experiments, e.g., non-convex probability or correlator inequalities for the analysis of quantum networks~\cite{bilocality_experimental, bilocality_experimental2, bilocality_experimental3, n-locality, Polino}, are no longer sound when the iid\ assumption is not satisfied. In fact, for any set of correlations $\S$, it is only possible to certify that the device is capable of producing correlations outside $\mbox{conv}(\S)$. In any  communication scenario, the convexification of set $\S$ can be interpreted as having all the communicating parties share an arbitrary source of classical randomness. By this reasoning, causal networks without inputs, such as the triangle scenario \cite{Fritz}, cannot exhibit any quantum advantage in the non-iid\ regime. Indeed, in any such network, a source of shared randomness allows all possible correlations to be achieved with classical resources.

Another consequence of Proposition~\ref{theo:main} is that, in general, it is impossible to certify that a source of quantum states is distributing separable states. Indeed, let ${\cal E}$ denote the set of entangled states and consider the separable two-qubit state $\rho=\frac{1}{2}(\proj{\psi^+}+\proj{\psi^-})$, with $\ket{\psi^{\pm}}:=\frac{1}{\sqrt{2}}(\ket{0}_A\ket{1}_B\pm \ket{1}_A\ket{0}_B)$. Since $\rho\in \mbox{conv}({\cal E})$, by Proposition \ref{theo:main} it is impossible to devise a test that outputs $0$ with high probability when the source is generating highly entangled states and $1$ when the source is distributing $n$ independent copies of $\rho$.

\section{Self testing without iid}\label{sec:self_test}

Let $\S$ be a set of physical correlations, and let $P\in\S$. Without relying on the iid assumption, under which circumstances can we guarantee that our device has prepared a distribution $Q$ such that $\|P-Q\|<\delta$? More formally, we consider the following problem.

\begin{problem}
\label{problem_st}
Given $\delta'> 0$, a closed set ${\cal P}$ of correlations and $P\in {\cal P}$, we aim to devise a family of hypothesis tests $(T^n)_n$ such that, for some $\delta\in(0,\delta')$,
\begin{equation}
\lim_{n\to\infty}P_{T^n}(1|Q^{\otimes n})=1,
\label{st_eII}
\end{equation}
for all $Q\in {\cal P}$ with $\|P-Q\|<\delta$. Moreover, there exists a sequence $(\epsilon_n)_n$ of real, non-negative numbers satisfying
\begin{equation}
\lim_{n\to\infty} \epsilon_n=0,
\label{limit_epsi}
\end{equation}
such that
\begin{align}
&P_{T^n}(1|P^{(n)})\leq \epsilon_n,\nonumber\\
\label{st_eI}
\end{align}
holds for all $P^{(n)}$ with marginals $\{P_k(a|x,s_k):k,s_k\}\subset {\cal P}$ satisfying $\|P_k(a|x,s_k)-P\|>\delta'$, $\forall k,s_k$.
\end{problem}
Not surprisingly, the answer of this problem depends on the relation between $P$ and $\mbox{conv}({\cal P})$.

\begin{prop}
\label{prop:st}
Problem \ref{problem_st} can be solved for arbitrarily small $\delta'>0$ iff $P$ is an extreme point of $\mbox{conv}({\cal P})$.
\end{prop}
\begin{proof}
First, we prove that the problem cannot be solved for all $\delta' >0$ if $P$ is not an extreme point of $\mbox{conv}({\cal P})$. Suppose, instead, that such is the case. By Caratheodory's theorem, there exist $\{\lambda_i\}_{i=1}^m\subset \R^+$, and $\{P_i\}_{i=1}^m$, with $m\leq \|P\|_0+1$ such that
\begin{equation}
P=\sum_{i=1}^m \lambda_iP_i,
\end{equation}
where $\|P\|_0$ denotes the number of non-zero probabilities of the behavior $P(a|x)$.
Now, suppose that there exists a family of binary hypothesis tests $(T^n)_n$ satisfying Eqs.~\eqref{st_eII}, (\ref{limit_epsi}) and (\ref{st_eI}) for $\delta'<\min_i\|P_i-P\|$ and arbitrary positive $\delta<\delta'$. Those tests comply with Eqs.~\eqref{non_local}, (\ref{eq:local_bound}) for $\S:=\{P_i\}_i$. Since $P\in\mbox{conv}(\S)\setminus \S$, their existence contradicts the statement of Proposition \ref{theo:main}.

Assume, on the contrary, that $P$ is an extreme point of $\mbox{conv}({\cal P})$, and, for arbitrarily small $\delta'>0$, define the closed set ${\cal A}:=\{Q\in \mbox{conv}({\cal P} ):\|Q-P\|\geq \delta'\}$. Since $P\not\in {\cal A}\subset \mbox{conv}({\cal P})$ and $P$ is extremal in $\mbox{conv}({\cal P})$, it follows that $P$ does not belong to the closed set $\mbox{conv}({\cal A})$. Therefore, by the Hahn-Banach theorem, there exist $\alpha\in\R$ and a separating linear functional $c$, which we can take to be normalized, such that $c(Q)\leq \alpha$, for all $Q\in \mbox{conv}({\cal A})$ and $c(P)>\alpha$. Setting $0<\delta< c(P)-\alpha$, we find that, for any $Q\in{\cal P}$ with $\|Q-P\|<\delta$, $c(Q)-\alpha\geq c(P)-\delta-\alpha>0$. To summarize, the linear functional $c$ satisfies $c(Q)\leq \alpha$, for all $Q\in{\cal P}$ with $\|Q-P\| \geq \delta'$ and $c(Q)>\alpha$, for all $Q\in{\cal P}$ with $\|Q-P\|<\delta$. Using the construction by \cite{ZhangPRA2013}, we can thus find a family of tests $(T^n)_n$ satisfying Eq.~\eqref{holy_grail} and thus the conditions of the proposition.
\end{proof}

A corollary of this proposition is that, in the non-iid regime, one cannot devise an experiment to self-test any physical operation $\tilde{O}$ that is not extreme in the set of operations allowed by the experimental setup, unless two or more uses of $\tilde{O}$ are allowed at each experimental round. 

Suppose, indeed, that we wished to certify that a given operation $O$ of our device is $\tilde{O}$, and let us assume that only one use of $O$ is allowed within an experimental round. By Bayes rule, this implies that, all other operations being fixed, the single-shot distribution $P(a|x,O)$ depends in a linear affine way on $O$.

Now, suppose that $\tilde{O}=\sum_i \lambda_i O_i$, where $\{O_i\}_i$ are extreme operations and the weights $\{\lambda_i\}_i\subset \R^+$ add up to $1$. Define $P_i:=P(a|x,O_i)$. Then, there is the possibility that $P_i=P$, for some $i$. In that case, one could not distinguish the non-extreme operation $\tilde{O}$ from the extreme $O_i$ with the considered setup. That is, it would be impossible to self-test $\tilde{O}$. Alternatively, it could be that $P_i\not=P$, for all $i$, in which case $P=\sum_i \lambda_iP_i$ would not be extreme in the set of allowed behaviors. Thus, by the above proposition, one cannot certify that the device generated a distribution arbitrarily close to $P$ in any experimental round, and so one cannot self-test $\tilde{O}$.

The impossibility of self-testing mixed states and non-extremal POVMs in Bell scenarios was already known \cite{baptista2023mathematicalfoundationselftestinglifting}. Self-testing has, however, been studied in communication scenarios other than Bell's, such as causal networks, where protocols to self-test arbitrary quantum states and measurements exist \cite{sarkar2023universal}. In a similar vein, self-testing results for mixed states and non-extremal measurements can be found for prepare-and-measure scenarios \cite{Nikolai}, also under the iid assumption. 

By our previous observation, though, any such protocol is vulnerable to memory attacks. A more positive consequence of Proposition \ref{prop:st} is that the work of \cite{Miguel} completely characterizes self-testing in the prepare-and-measure scenario (in the non-iid regime). There the authors provide a general protocol to self-test any ensemble of extreme (pure) quantum preparations and extreme POVMs. By Proposition \ref{prop:st}, no other states and measurements can be self-tested.

What happens if we allow more than one use of $O$ within a single experimental round?  That would be a form of iid \emph{within} rounds, rather than \emph{between} rounds. It might be a good approximation in many experimental setups, though, if all the calls to $O$ take place in a very short time span. 

If we allow, say, up to two uses of $O$ per round, then each of the probabilities of the single-shot behavior $P(a|x,O)$ will be a second-degree polynomial on the parameters defining $O$. Equivalently, it will be an affine linear expression on $O$ and $O\otimes O$. Hence, the relevant set to consider is not the set ${\cal O}$ of all physical operations, but the set ${\cal O}^2:=\mbox{conv}\left(\{(O,O\otimes O):O\in{\cal O}\}\right)$. While $\tilde{O}\in{\cal O}$ might not be an extreme point of ${\cal O}$, it is possible that $(\tilde{O},\tilde{O}^{\otimes 2})$ is an extreme point of ${\cal O}^2$, in which case the no-go argument above does not apply.

Suppose, for instance, that ${\cal O}$ is the set ${\cal Q}_D$ of $D$-dimensional quantum states. Then, ${\cal O}^2$ would correspond to the set ${\cal Q}_D^2$ of pairs of matrices that can be expressed as convex combinations of pairs of the form $(\rho, \rho^{\otimes 2})$, where $\rho$ is a quantum state. It turns out that, for \emph{any} $\rho\in {\cal Q}_D$, $(\rho,\rho^{\otimes 2})$ is an exposed, extreme point of ${\cal Q}^2_D$. Indeed, consider the following $D^2\times D^2$ matrix:
\begin{equation}
W_\rho=V+\tr(\rho^2)\id_D^{\otimes 2}-2\id_D\otimes \rho,
\end{equation}
where $V\ket{j}\ket{k}=\ket{k}\ket{j}$ denotes the permutation operator. This matrix gives rise to the linear functional ${\cal W}_\rho:{\cal Q}_D^2\to\R$ through the relation ${\cal W}_\rho((\alpha,\beta)):=\tr\{W_\rho \beta\}$, for $(\alpha,\beta)\in{\cal Q}_D^2$. It can be verified that, for any $\sigma\in{\cal Q}_D$, ${\cal W}_\rho((\sigma,\sigma^{\otimes 2}))=\tr\{(\sigma-\rho)^2\}$. Thus, the linear optimization problem 
\begin{align}
&\min{\cal W}_\rho(q)\nonumber\\
\mbox{such that }&q\in{\cal Q}_D^2
\end{align}
has $q=(\rho,\rho^{\otimes 2})$ as a unique solution, and so $(\rho,\rho^{\otimes 2})$ is an exposed point of ${\cal Q}_D^2$. This argument can be extended to any set ${\cal O}$ that can be embedded in a countable vector space and such that there exists a linear functional $\omega:{\cal O}\to\R$ with $\omega(O)=1$, for all $O\in {\cal O}$. This includes the set of all $D$-dimensional POVMs/quantum channels/quantum instruments. 

 Hence, Proposition \ref{prop:st} does not preclude the possibility of self-testing, e.g., all $D$-dimensional POVMs when more than a single measurement is allowed within a given round. These ideas have been investigated by \cite{Das_2022,classicalcertificationquantumgates}. Indeed, the sequential measurement protocol proposed in \cite{Das_2022} does precisely that: it self-tests arbitrary (i.e., non-necessarily extreme) $D$-dimensional POVMs through a witness linear in the input/output distribution $P(a_1,a_2|x_1,x_2)$.

\section{Prior communication in causal network experiments}
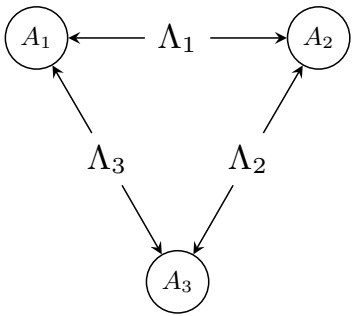
\begin{figure}[t]
	\centering
	\resizebox{0.65\columnwidth}{!}{%
		\begin{tikzpicture}[scale=0.65]    \node[draw=black,circle,scale=0.75]  (X) at (-2,2) {$A_1$};
        \node[draw=black,circle,scale=0.75]  (Y) at (2,2) {$A_2$};
	    \node[draw=black,circle,scale=0.75]  (Z) at (0,-1.46) {$A_3$};
		\node (A) at (1,0.28) {$\Lambda_2$};
			\node (B) at (-1,0.28) {$\Lambda_3$};
			\node (C) at (0,2) {$\Lambda_1$};
			\draw [->,>=stealth] (A)--(Y) node [right,pos=0.4] {};
			\draw [->,>=stealth] (A)--(Z) node [right,pos=0.5] {};
			\draw [->,>=stealth] (B)--(X) node [left,pos=0.4] {};
			\draw [->,>=stealth] (B)--(Z) node [left,pos=0.5] {};
			\draw [->,>=stealth] (C)--(X) node [above,pos=0.5] {};
			\draw [->,>=stealth] (C)--(Y) node [above,pos=0.5] {};
		\end{tikzpicture}
	}%
	\caption{Triangle network. Three parties, Alice, Bob and Charlie, make observations captured by random variables $A_1,A_2,A_3$ respectively. In the example in the main text these are binary, i.e., take values in $\{0,1\}$. For this, each pair of parties shares a bipartite source of shared randomness, in terms of a random variable $\Lambda_1, \Lambda_2, \Lambda_3$ respectively (see diagram). In the quantum case these are replaced by quantum states, which increases the number of distributions $P(a_1,a_2,a_3)$ that can be generated.}
	\label{fig:triangle}
\end{figure}

Consider the triangle causal network~\cite{Steudel, Fritz}, depicted in Figure~\ref{fig:triangle}. In the triangle scenario three separate agents locally generate random variables by combining three bipartite sources. Thus, a tri-variate distribution $P(a_1,a_2,a_3)$ is compatible with the triangle network if there exist distributions $p_1,p_2,p_3, q_1, q_2,q_3$ such that
\begin{equation}\label{eq:P_triangle}
\begin{split}
P(a_1,a_2,a_3)=\sum_{\lambda_1,\lambda_2,\lambda_3} & p_1(\lambda_1)p_2(\lambda_2)p_3(\lambda_3) q_1(a_1|\lambda_1,\lambda_3)\\
&q_2(a_2|\lambda_1,\lambda_2)q_3(a_3|\lambda_2,\lambda_3).
\end{split}
\end{equation}
Call $\S$ the set of all single-site distributions satisfying the decomposition above. Remarkably, some tri-variate distributions cannot be arbitrarily well approximated by any distribution in $\S$. One of these is the  distribution $P_c(0,0,0)=P_c(1,1,1)= \frac{1}{2}$, which violates a number of non-linear inequalities satisfied by all triangle-feasible distributions~\cite{Steudel,Fritz,NonShannon,Inflation, Chaves_2014}. 

The reader might wonder what would happen if we ran a family of tests $(T^n)_n$, tailor-made to detect $P_c\in \S$ under the iid assumption, but otherwise allowed the three triangle-constrained players to use non-iid strategies. By Proposition \ref{theo:main}, it follows that $(T^n)_n$ cannot satisfy Eq. (\ref{eq:local_bound}). Thus, there must exist non-iid behaviors $P^{(n)}\in \S^{\otimes n}$ that allow the parties to pass $T^n$ with non-negligible probability in the limit $n\to\infty$. But, what do these behaviors look like?

As it turns out, for any such family of tests there exist many feasible behaviors that asymptotically pass the test with probability tending to $1$. Moreover, some of them do not depend on the specific single-shot witness function $F$ employed to decide the outcome, see Eq.~\eqref{eq:prob_ineq}. One of these behaviors is the `clock strategy' $P^{(n)}_{clock}:=(P_0\otimes P_1)^{\otimes n/2}$, with $P_b(a_1,a_2,a_3):=\delta_{b,a_1}\delta_{b,a_2}\delta_{b,a_3}$. For $n=2m\geq 2$, the observed average frequencies $\tilde{P}(a_1,a_2,a_3)$, see \eqref{eq:prob_est_iid}, would satisfy $\tilde{P}=P_c$ exactly and so the three parties would pass the hypothesis test.

The clock strategy requires each player $i$ to keep track of an internal variable $a_i\in\{0,1\}$, which tells them what digit to output and gets updated at each round. In this case, the update rule is very simple: $a_i\to a_i\oplus 1$. The scheme demands the three players to have the same values of the variables ($a_1,a_2,a_3)$ at round $1$: should player $1$ not be synchronized with the other two, the frequency average would change to $\tilde{P}(1,0,0)=\tilde{P}(0,1,1)=\frac{1}{2}$. This single-shot distribution is also unrealizable within the triangle network, but might not maximize functional $F$.

The reader might not think remarkable that three independent binary counters coincide. In fact, it is easy to find natural mechanisms to ensure this, e.g.: the memory system storing $a_i$ might decay to state $a_i=0$ in the absence of frequent updates. That all three players also independently come up with the same local update rule is less clear, but may still seem plausible. Consider, though, the following strategy $R^{(n)}$: let $q$ be a $n$-bit random sequence, sampled from a uniform distribution. Then, at round $k$ each party just needs to output $q_k$ to generate $\tilde{P}\approx P_c$ with high probability. At the same time, with high probability no $n$-round test can distinguish this \emph{deterministic} strategy from $P_c^{\otimes n}$. 

The behavior $R^{(n)}$ is deterministic and thus compatible with the triangle scenario. However, it does not seem \emph{plausible} 
that three separate parties happen to come up with the same sequence $q$ of numbers out of the blue, even though it is in principle possible. Coordinating strategies of this sort requires communication between the agents \emph{before} the test starts.

To our knowledge, the role of prior communication in causal network experiments has not been discussed, except for comments on the role of \emph{shared reference frames} in specific examples~\cite{reference_frames, with_inputs_1}. This is most likely because this discussion is irrelevant both in the iid\ regime and in any causal network with a common source of randomness, such as a Bell test or the recent experimental disproofs of real quantum mechanics \cite{Renou_2021, RQ_exp1, RQ_exp2, RQ_exp3}. 
In non-iid\ scenarios without common classical sources, though, prior communication makes a difference. Depending on the goal of the subsequent test, it might be completely forbidden, but also partially or fully allowed, as long as it is accounted for. We next examine these three possibilities. In each case, we will take for granted that all the parties are informed of the test they are about to take. This premise is characteristic of cryptographic scenarios, where the protocol to be implemented is assumed to `leak' to the public.

Unlimited prior communication, while harmless in the iid\ case, would in general allow separate parties to run strategies using unlimited amounts of shared randomness, such as $R^{(n)}$. In such a scenario, the set of feasible behaviors is indistinguishable from those generated by a network with a common randomness source.

Other scenarios might allow a limited amount of prior communication between the parties, before or after the test is leaked. In this regard, a minimal paradigm would allow some known, prior, bounded, test-independent knowledge common to all the parties involved in the communication game, e.g.\ all the parties speak English. 

This paradigm would preclude the implementation of strategy $R^{(n)}$, for $n\gg 1$. It is, however, sufficient to run $P^{(n)}_{clock}$; in fact, it allows the players to pass any hypothesis test with maximum probability. In effect, all they need to do is follow this meta-strategy: once they receive the description of the game, each party (locally) computes all the optimal deterministic strategies to win it. Among those, they follow the first optimal strategy in lexicographic order (here we need to demand that the three players know the meta-strategy, and hence the notion of lexicographic order). This procedure guarantees that, no matter the test, the three parties will implement an optimal strategy, not only among all deterministic ones, but in general. This is so due to the following trivial identity:
\begin{equation}
\max_{P^{(n)}\in \mathbb{P}} P_T(1|P^{(n)})=\max_{P^{(n)}\in \mathbb{P}_{det}} P_T(1|P^{(n)}),
\end{equation}
where $\mathbb{P}$ ($\mathbb{P}_{det}$) denotes the set of all (deterministic) $n$-round behaviors. This meta-strategy can be straightforwardly generalized to any classical causal network without inputs. 

Finally, a third experimental scenario would ban any form of communication between the parties prior to  the test, other than the knowledge of the test itself. This prescription would rule out all multipartite quantum experiments conducted so far on causal networks without a common source of randomness, even under the iid assumption. Indeed, those experiments, which require all the parties to conduct appropriate quantum operations, are impossible to set up without contacting all the parties involved. Furthermore,  experiments that involve independent sources generally also require a synchronisation of the sources in different rounds, which according to this viewpoint would also be considered part of the (forbidden) communication.

\section{Conclusions}

In this paper, we have proven that basic primitives of quantum information theory cannot be implemented when one drops the iid assumption. More specifically, we showed that, in the non-iid regime, one cannot devise: (1) binary hypothesis tests to decide membership within a non-convex set and (2) self-tests of non-extreme operations (as long as we only allow the use of one operation per experimental round).

As a consequence of (1), currently used membership criteria for nonconvex sets, typically associated with causal network inequalities, do not satisfy the requirements of a general hypothesis test for the non-iid scenario. Moreover, the non-classicality of causal networks without inputs, such as the triangle network, cannot be certified. A consequence of (2) is that existing self-testing protocols for mixed density matrices and quantum measurements \cite{Nikolai, sarkar2023universal} cannot resist memory attacks. On a more positive note, (2) also implies that the recent result by one of us \cite{Miguel} completely characterizes self-testing in prepare-and-measure scenarios. 

While analyzing the form of optimal $n$-round strategies in network communication tests, we raised the need to account for prior communication between the parts of the experiment involved. As we saw, given unrestrained communication, the statistics generated by the parties are indistinguishable from those corresponding to a causal network with a common classical source node. In fact, just knowing the test to be conducted is enough for the players to `simulate' the presence of such an extra node. On the other hand, banning prior communication altogether would not allow setting up most feasible iid strategies. Ultimately, deciding what is or not allowed is determined by the practical use we want to make of the network.

\begin{acknowledgements}

\begin{wrapfigure}{r}[0cm]{2cm}
\begin{center}
\includegraphics[width=2cm]{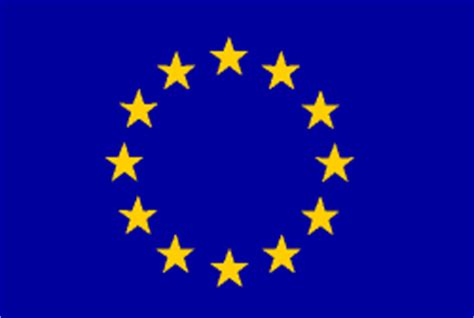}    
\end{center}
\end{wrapfigure} 

This project was supported by the Swiss National Science Foundation (Ambizione PZ00P2$\_$208779); by the QuantERA II Programme that has received funding from the European Union's Horizon 2020 research and innovation programme under Grant Agreement No 101017733 and by the Austrian Science Fund (FWF) through project I-6004, project ZK 3, and project F 7113.

\vspace{.5cm}
    
\end{acknowledgements}

\bibliography{iid}

\end{document}